\documentclass{eptcs}

\providecommand{\event}{Quantum Physics and Logic 2017}

\input{cnotdihedral.sty} % U triangle.
\def\Utriangle{
\m{\begin{qcircuit}[scale=0.39]
    \grid{2.8}{0,3.2,4.8}
    \utriangle{1.4}{0}{4.8}{1.6,4.8}
\end{qcircuit}}
=
\m{\begin{qcircuit}[scale=0.39]
    \Period{7.6}{0}
    \grid{7.2}{0,3.2,4.8}
    \Ugate{n_1}{1.4}{4.8}{3.2}
    \colgate{white, white}{$\bcdots$}{3.68,4.8}
    \colgate{white, white}{$\bcdots$}{3.68,3.2}
    \colgate{white, white}{$\bddots$}{3.68,1.6}
    \colgate{white, white}{$\bcdots$}{3.68,0}
    \Ugate{n_k}{5.8}{4.8}{0}
\end{qcircuit}}
}

\def\Vtrapezoid{
\m{\begin{qcircuit}[scale=0.39]
    \grid{2.8}{0,3.2,4.8,6.4}
    \vtrapezoid{1.4}{6.4}{4.8}{0}{2,1}
\end{qcircuit}}
=
\m{\begin{qcircuit}[scale=0.39]
    \grid{7.2}{0,3.2,4.8,6.4}
    \Vgate{n_1}{1.4}{3.2}{6.4}{4.8}
    \colgate{white, white}{$\bcdots$}{3.68,6.4}
    \colgate{white, white}{$\bcdots$}{3.68,4.8}
    \colgate{white, white}{$\bcdots$}{3.68,3.2}
    \colgate{white, white}{$\bddots$}{3.68,1.6}
    \colgate{white, white}{$\bcdots$}{3.68,0}
    \Vgate{n_k}{5.8}{4.8}{6.4}{0}
\end{qcircuit}}
}

\def\Vtriangle{
\m{\begin{qcircuit}[scale=0.39]
    \grid{2.8}{0,1.6,4.8,6.4}
    \vtriangle{1.4}{0}{6.4}{1.6,6.4}
\end{qcircuit}}
=
\m{\begin{qcircuit}[scale=0.39]
    \Period{7.6}{0}
    \grid{7.2}{0,1.6,4.8,6.4}
    \vtrapezoid{1.4}{6.4}{4.8}{0}{2,1}
    \colgate{white, white}{$\bcdots$}{3.68,6.4}
    \colgate{white, white}{$\bcdots$}{3.68,4.8}
    \colgate{white, white}{$\bddots$}{3.68,3.2}
    \colgate{white, white}{$\bcdots$}{3.68,1.6}
    \colgate{white, white}{$\bcdots$}{3.68,0}
    \vtrapezoid{5.8}{6.4}{1.6}{0}{2,1}
\end{qcircuit}}
}

\def\stair{
\m{\begin{qcircuit}[scale=0.39]
    \grid{5}{0,1.6,3.2,6.4,8}
    \colgate{white, white}{$\bvdots$}{0.8,4.8}
    \colgate{white, white}{$\bvdots$}{4.2,4.8}
    \staircasealt{1.4}{0}{3.6}{8}{2.5,4}
\end{qcircuit}}
=
\m{\begin{qcircuit}[scale=0.39]
    \Period{16.4}{0}
    \grid{16}{0,1.6,3.2,6.4,8}
    \colgate{white, white}{$\bvdots$}{1.4,4.8}
    \colgate{white, white}{$\bcdots$}{10.2,8}
    \colgate{white, white}{$\bcdots$}{10.2,6.4}
    \colgate{white, white}{$\bcdots$}{10.2,4.8}
    \colgate{white, white}{$\biddots$}{10.2,4.8}
    \colgate{white, white}{$\bcdots$}{10.2,3.2}
    \colgate{white, white}{$\bcdots$}{10.2,1.6}
    \colgate{white, white}{$\bcdots$}{10.2,0}
    \Swap{1.4}{0}{1.6}
    \CX{n_1}{3.6}{1.6}{0}
    \Swap{5.8}{1.6}{3.2}
    \CX{n_2}{8}{3.2}{1.6}
    \Swap{12.4}{8}{6.4}
    \CX{n_k}{14.6}{8}{6.4}
    \colgate{white, white}{$\bvdots$}{14.6,4.8}
\end{qcircuit}} 
}

\def\ladder{
\m{\begin{qcircuit}[scale=0.39]
    \grid{2.8}{0,3.2,6.4}
    \colgate{white, white}{$\bvdots$}{0.4,1.6}
    \colgate{white, white}{$\bvdots$}{0.4,4.8}
    \colgate{white, white}{$\bvdots$}{2.4,1.6}
    \colgate{white, white}{$\bvdots$}{2.4,4.8}
    \updownalt{1.4}{0}{1.4}{6.4}{1.4}{0}{1.4,0.9}
\end{qcircuit}}
=
\m{\begin{qcircuit}[scale=0.39]
    \Period{8.7}{0}
    \grid{8.3}{0,3.2,6.4}
    \colgate{white, white}{$\bvdots$}{0.4,1.6}
    \colgate{white, white}{$\bvdots$}{0.4,4.8}
    \staircasealt{1.4}{0}{3.6}{6.4}{2.5,3.2}
    \colgate{white, white}{$\bvdots$}{7.9,1.6}
    \colgate{white, white}{$\bvdots$}{7.9,4.8}
    \staircasealt{6.9}{3.2}{5.8}{6.4}{6.3,4.8}
\end{qcircuit}}
}

\def\nfaff{
\m{\begin{qcircuit}[scale=0.39]
    \grid{11.6}{0,3.2,4.8}
    \updownalt{1.4}{4.8}{1.4}{4.8}{1.4}{4.8}{1.37,4.65}
    \updownalt{3.6}{3.2}{3.6}{4.8}{3.6}{3.2}{3.55,3.2}
    \colgate{white, white}{$\bcdots$}{5.8,4.8}
    \colgate{white, white}{$\bcdots$}{5.8,3.2}
    \colgate{white, white}{$\bddots$}{5.8,1.6}
    \colgate{white, white}{$\bcdots$}{5.8,0}
    \updownalt{8}{0}{8}{4.8}{8}{0}{8,0.8}
    \Xblock{}{10.2}{4.8}
    \Xblock{}{10.2}{3.2}
    \Xblock{}{10.2}{0}
    \colgate{white, white}{$\bvdots$}{10.2,1.6}
  \end{qcircuit}}
}

\def\exampleAff{
\m{\begin{qcircuit}[scale=0.39]
    \Period{21}{0}
    \grid{20.6}{0,1.6,3.2}
    \Swap{1.4}{3.2}{1.6}
    \Swap{3.6}{0}{1.6}
    \Swap{5.8}{1.6}{3.2}
    \CX{}{8}{3.2}{1.6}
    \Swap{10.2}{1.6}{3.2}
    \CX{}{12.4}{3.2}{1.6}
    \Swap{14.6}{0}{1.6}
    \CX{}{16.8}{1.6}{0}
    \Xgate{}{19}{3.2}
\end{qcircuit}} 
}

\def\nfdiag{
\m{\begin{qcircuit}[scale=0.39]
    \Omeg{k}{-1}{2.4} 
    \grid{20.4}{0,3.2,4.8}
    \Tblock{}{1.4}{4.8}
    \colgate{white, white}{$\bvdots$}{1.4,1.6}
    \Tblock{}{1.4}{3.2}
    \Tblock{}{1.4}{0}
    \utriangle{3.6}{0}{4.8}{3.75,4.8}
    \utriangle{5.8}{0}{3.2}{5.95,3.2}
    \colgate{white, white}{$\bcdots$}{8.08,4.8}
    \colgate{white, white}{$\bddots$}{8.08,1.6}
    \colgate{white, white}{$\bcdots$}{8.08,3.2}
    \colgate{white, white}{$\bcdots$}{8.08,0}
    \utriangle{10.2}{0}{0}{10.65,0.2}
    \vtriangle{12.4}{0}{4.8}{12.55,4.8}
    \vtriangle{14.6}{0}{3.2}{14.75,3.2}
    \colgate{white, white}{$\bcdots$}{16.88,4.8}
    \colgate{white, white}{$\bddots$}{16.88,1.6}
    \colgate{white, white}{$\bcdots$}{16.88,3.2}
    \colgate{white, white}{$\bcdots$}{16.88,0}
    \vtriangle{19}{0}{0}{19.45,0.2}
  \end{qcircuit}}
}

\def\nfCCZ{
\m{\begin{qcircuit}[scale=0.39]
    \grid{2.8}{0,1.6,3.2}
    \Ccz{1.4}{0}{1.6}{3.2}
\end{qcircuit}}
=
\m{\begin{qcircuit}[scale=0.39]
    \Period{12}{0}
    \grid{11.6}{0,1.6,3.2}
    \Tgate{}{1.4}{0}
    \Tgate{}{1.4}{1.6}
    \Tgate{}{1.4}{3.2}
    \Ugate{3}{3.6}{1.6}{3.2}
    \Ugate{3}{5.8}{0}{3.2}
    \Ugate{3}{8}{0}{1.6}
    \Vgate{}{10.2}{0}{1.6}{3.2}
\end{qcircuit}}
}

\def\bifunctoriality{
\m{\begin{qcircuit}[scale=0.39]
    \grid{5}{0,1.6}
    \colgate{white!20}{$f$}{1.4,1.6}
    \colgate{white!20}{$g$}{3.6,0}
\end{qcircuit}}
=
\m{\begin{qcircuit}[scale=0.39]
    \Period{5.3}{0}
    \grid{4.9}{0,1.6}
    \colgate{white!20}{$g$}{1.4,0}
    \colgate{white!20}{$f$}{3.6,1.6}
\end{qcircuit}}
}

\def\symmetries{
\m{\begin{qcircuit}[scale=0.39]
    \grid{2.8}{0,1.6}
    \Swap{1.4}{0}{1.6}
\end{qcircuit}}
\qquad\mbox{ and }\qquad
\m{\begin{qcircuit}[scale=0.39] 
    \grid{2.8}{0,1.6,3.2}
    \cleargate{$T$}{1.4,0}
    \cleargate{$T$}{1.4,1.6}
    \cleargate{$T$}{1.4,3.2} 
    \diagwire{1.4}{0}{3.2} 
    \diagwire{1.4}{3.2}{1.6} 
    \diagwire{1.4}{1.6}{0}
\end{qcircuit}}
}

\def\Naturality{
\m{\begin{qcircuit}[scale=0.39]
    \grid{5}{0,1.6}
    \colgate{white!20}{$f$}{1.4,1.6}
    \colgate{white!20}{$g$}{1.4,0}
    \Swap{3.6}{0}{1.6}
\end{qcircuit}}
=
\m{\begin{qcircuit}[scale=0.39]
    \grid{5}{0,1.6}
    \colgate{white!20}{$g$}{3.6,1.6}
    \colgate{white!20}{$f$}{3.6,0}
    \Swap{1.4}{0}{1.6}
\end{qcircuit}}
\qquad\mbox{ and }\qquad
\m{\begin{qcircuit}[scale=0.39] 
    \grid{5}{0,1.6,3.2}
    \widebigcolgate{$h$}{1.4,1.6}{1.4,3.2}{.8}{white}
    \cleargate{$T$}{3.6,0}
    \cleargate{$T$}{3.6,1.6}
    \cleargate{$T$}{3.6,3.2} 
    \diagwire{3.6}{0}{3.2} 
    \diagwire{3.6}{3.2}{1.6} 
    \diagwire{3.6}{1.6}{0}
    \colgate{white!20}{$f$}{1.4,0}
\end{qcircuit}}
=
\m{\begin{qcircuit}[scale=0.39] 
    \Period{5.4}{0}
    \grid{5}{0,1.6,3.2}
    \colgate{white!20}{$f$}{3.6,3.2}
    \widebigcolgate{$h$}{3.6,0}{3.6,1.6}{.8}{white}
    \cleargate{$T$}{1.4,0}
    \cleargate{$T$}{1.4,1.6}
    \cleargate{$T$}{1.4,3.2} 
    \diagwire{1.4}{0}{3.2} 
    \diagwire{1.4}{3.2}{1.6} 
    \diagwire{1.4}{1.6}{0}
\end{qcircuit}}
}

\def\spatiality{ 
\m{\begin{qcircuit}[scale=0.39]
    \grid{2.8}{0}
    \colgate{white!20}{$\lambda$}{1.4,1.6}
\end{qcircuit}}
=
\m{\begin{qcircuit}[scale=0.39]
    \grid{2.8}{1.6}
    \colgate{white!20}{$\lambda$}{1.4,0}
\end{qcircuit}}
}

\def\coherence{
\m{\begin{qcircuit}[scale=0.39]
    \grid{5}{0,1.6,3.2}
    \Swap{1.4}{0}{1.6}
    \Swap{3.6}{1.6}{3.2}
\end{qcircuit}}
=
\m{\begin{qcircuit}[scale=0.39] 
    \Period{3.2}{0}
    \grid{2.8}{0,1.6,3.2}
    \cleargate{$T$}{1.4,0}
    \cleargate{$T$}{1.4,1.6}
    \cleargate{$T$}{1.4,3.2} 
    \diagwire{1.4}{0}{3.2} 
    \diagwire{1.4}{3.2}{1.6} 
    \diagwire{1.4}{1.6}{0}
\end{qcircuit}}
}

\def\block{
\m{\begin{qcircuit}[scale=0.39] \grid{2.8}{0}
    \colgate{white!20}{$f^n$}{1.4,0}
\end{qcircuit}} 
=
\mp{0.75}{\begin{qcircuit}[scale=0.39] 
    \colgate{white, white}{$\bdot$}{7.5,0} 
    \colgate{white, white}{$\bcdots$}{3.67,0} 
    \grid{2.8}{0} 
    \gridx{4.4}{7.2}{0} 
    \colgate{white!20}{$f$}{1.4,0} 
    \colgate{white!20}{$f$}{5.8,0} 
    \brace{$n$}{0.6}{6.6}{-1}
\end{qcircuit}}
}

\def\NonAdjacent{
\m{\begin{qcircuit}[scale=0.39] 
    \grid{2.8}{0,1.6,3.2}
    \controlwires{1.4,0}{3.2}
    \colgate{white}{$f$}{1.4,0} 
    \colgate{white}{$f$}{1.4,3.2}
\end{qcircuit}} 
= 
\m{\begin{qcircuit}[scale=0.39]
    \Period{7.6}{0}
    \grid{7.2}{0,1.6,3.2} 
    \Swap{1.4}{1.6}{0}
    \controlwires{3.6,1.6}{3.2}
    \colgate{white}{$f$}{3.6,1.6} 
    \colgate{white}{$f$}{3.6,3.2}
    \Swap{5.8}{1.6}{0}
\end{qcircuit}}
}

\def\DerivationOne{
\begin{align*}
    \m{\begin{qcircuit}[scale=0.39]
        \grid{5}{0,1.6}
        \Swap{1.4}{0}{1.6}
        \Ugate{}{3.6}{0}{1.6}
      \end{qcircuit}}
    &= \m{\begin{qcircuit}[scale=0.39]
        \grid{18.2}{0,1.6}
        \CX{}{1.4}{1.6}{0}
        \Swap{3.6}{1.6}{0}
        \CX{}{5.8}{1.6}{0}
        \Swap{8}{1.6}{0}
        \CX{}{10.2}{1.6}{0}
        \CX{}{12.4}{1.6}{0}
        \Tgate{}{14.6}{0}
        \CX{}{16.8}{1.6}{0}
      \end{qcircuit}} \\
    &= \m{\begin{qcircuit}[scale=0.39]
        \grid{13.8}{0,1.6}
        \CX{}{1.4}{1.6}{0}
        \Swap{3.6}{1.6}{0}
        \CX{}{5.8}{1.6}{0}
        \Swap{8}{1.6}{0}
        \Tgate{}{10.2}{0}
        \CX{}{12.4}{1.6}{0}
      \end{qcircuit}} \\
    &= \m{\begin{qcircuit}[scale=0.39]
        \grid{13.8}{0,1.6}
        \CX{}{1.4}{1.6}{0}
        \Swap{3.6}{1.6}{0}
        \CX{}{5.8}{1.6}{0}
        \Tgate{}{8}{1.6}
        \Swap{10.2}{1.6}{0}
        \CX{}{12.4}{1.6}{0}
      \end{qcircuit}} \\
    &= \m{\begin{qcircuit}[scale=0.39]
        \grid{13.8}{0,1.6}
        \CX{}{1.4}{1.6}{0}
        \Swap{3.6}{1.6}{0}
        \Tgate{}{5.8}{1.6}
        \CX{}{8}{1.6}{0}
        \Swap{10.2}{1.6}{0}
        \CX{}{12.4}{1.6}{0}
      \end{qcircuit}} \\
    &= \m{\begin{qcircuit}[scale=0.39]
        \grid{13.8}{0,1.6}
        \CX{}{1.4}{1.6}{0}
        \Tgate{}{3.6}{0}
        \Swap{5.8}{1.6}{0}
        \CX{}{8}{1.6}{0}
        \Swap{10.2}{1.6}{0}
        \CX{}{12.4}{1.6}{0}
      \end{qcircuit}} \\
    &= \m{\begin{qcircuit}[scale=0.39]
        \grid{18.2}{0,1.6}
        \CX{}{1.4}{1.6}{0}
        \Tgate{}{3.6}{0}
        \CX{}{5.8}{1.6}{0}
        \CX{}{8}{1.6}{0}
        \Swap{10.2}{1.6}{0}
        \CX{}{12.4}{1.6}{0}
        \Swap{14.6}{1.6}{0}
        \CX{}{16.8}{1.6}{0}
      \end{qcircuit}}
     = \m{\begin{qcircuit}[scale=0.39]
        \Period{5.4}{0}
        \grid{5}{0,1.6}
        \Ugate{}{1.4}{0}{1.6}
        \Swap{3.6}{0}{1.6}
      \end{qcircuit}}
\end{align*}
}

\def\DerivationTwo{
\begin{align*}
\m{\begin{qcircuit}[scale=0.39]
    \grid{13.8}{0,1.6,3.2,4.8}
    \Ugate{}{1.4}{3.2}{4.8}
    \CX{}{3.6}{3.2}{1.6}
    \CX{}{5.8}{1.6}{0}
    \Tgate{}{8.0}{0}
    \CX{}{10.2}{1.6}{0}
    \CX{}{12.4}{3.2}{1.6}
\end{qcircuit}}
=
\m{\begin{qcircuit}[scale=0.39]
    \grid{13.8}{0,1.6,3.2,4.8}
    \CX{}{1.4}{3.2}{1.6}
    \CX{}{3.6}{1.6}{0}
    \Tgate{}{5.8}{0}
    \CX{}{8.0}{1.6}{0}
    \CX{}{10.2}{3.2}{1.6}
    \Ugate{}{5.8}{3.2}{4.8}
\end{qcircuit}}
=
\m{\begin{qcircuit}[scale=0.39]
    \Period{14.2}{0}
    \grid{13.8}{0,1.6,3.2,4.8}
    \CX{}{1.4}{3.2}{1.6}
    \CX{}{3.6}{1.6}{0}
    \Tgate{}{5.8}{0}
    \CX{}{8.0}{1.6}{0}
    \CX{}{10.2}{3.2}{1.6}
    \Ugate{}{12.4}{3.2}{4.8}
\end{qcircuit}}
\end{align*}
}

\title{A Finite Presentation of $\mathbf{CNOT}$-Dihedral
  Operators}

\author{
	Matthew Amy
	\institute{Institute for Quantum Computing and \\
		David R. Cheriton School of Computer Science \\ 
		University of Waterloo \\ 
		Waterloo, Canada}
	\email{matt.e.amy@gmail.com}
	\and
	Jianxin Chen
	\qquad\qquad
  	Neil J.\ Ross
  	\institute{Institute for Advanced Computer Studies and \\
  		Joint Center for Quantum Information and Computer Science \\ 
  		University of Maryland \\
  		College Park, USA}
	\email{chenkenshin@gmail.com \qquad neil.jr.ross@gmail.com}
}

\def\titlerunning{A Finite Presentation of $\CNOT$-Dihedral
  Operators}
\def\authorrunning{M. Amy, J. Chen \& N. J. Ross}

\begin{document}

\maketitle

\begin{abstract}
  We give a finite presentation by generators and relations of the unitary
  operators expressible over the $\cnottx$ gate set, also known as
  $\CNOT$-dihedral operators. To this end, we introduce a notion of
  normal form for $\CNOT$-dihedral circuits and prove that every
  $\CNOT$-dihedral operator admits a unique normal form. Moreover, we
  show that in the presence of certain structural rules only finitely
  many circuit identities are required to reduce an arbitrary
  $\CNOT$-dihedral circuit to its normal form.

  By appropriately restricting our relations, we obtain a finite
  presentation of unitary operators expressible over the $\cnott$ gate
  set as a corollary.
\end{abstract}

% --------------------------------------------------------------------
\section{Introduction}
\label{sec:intro}

The \emph{Clifford+$T$} gate set consists of the $\CNOT$, Hadamard,
and $T$ gates \cite{NC}. This gate set has been the focus of recent
efforts in the study of quantum circuits due to its close connection
to quantum fault tolerance.  As a result, the theory of single-qubit
Clifford+$T$ circuits is now well-established
\cite{KMM-exact,MA08,RS16}. In contrast, multi-qubit Clifford+$T$
circuits are not very well understood, despite interesting results
\cite{GS13,GKMR}. The difficulties associated with multi-qubit
circuits shifted emphasis from the full Clifford+$T$ gate set to
restricted classes of circuits. In particular, circuits over the
$\s{\CNOT, T, X}$ gate set, known as \emph{$CNOT$-dihedral circuits of
  order 16}\footnote{Circuits over the $\s{\CNOT,T,X}$ are known as
  $\CNOT$-dihedral circuits of order 16 because the group generated by
  $T$ and $X$ is isomorphic to the dihedral group of order 16
  \cite{CMBSG16}. For brevity, we omit the order of the associated
  dihedral group and refer to $\s{\CNOT, T, X}$ circuits as
  $\CNOT$-dihedral circuits.}, and circuits over the $\s{\CNOT, T}$
gate set, known as \emph{CNOT+$T$ circuits}, received significant
attention. This led to a randomized benchmarking procedure for
$\CNOT$-dihedral circuits \cite{CMBSG16} as well as circuit
optimizations \cite{AMM,AMMR,AM} and improved distillation protocols
\cite{CH161} for $\CNOT$+$T$ circuits.

We give a finite presentation of $\CNOT$-dihedral operators of order
16 in terms of generators and relations, inspired by similar results
given for Clifford operators in \cite{Sel} and certain classes of
Boolean operators in \cite{Laf}. First, we introduce normal forms for
$\CNOT$-dihedral circuits. Then, we prove that, in the presence of
certain structural rules described in \cref{sec:prelims}, a finite set
of circuit equalities (the relations) suffices to reduce an arbitrary
$\CNOT$-dihedral circuit to its normal form. This shows that normal
form representations of $\CNOT$-dihedral operators always
exist. Finally, we show that distinct normal forms represent distinct
operators, which implies that normal form representations are unique.
These results yield a presentation by generators and relations of the
collection of $\CNOT$-dihedral operators as a symmetric monoidal
groupoid (see \cref{sec:prelims} for more details). By restricting the
generators and relations from $\s{\CNOT, T, X}$ to $\s{\CNOT, T}$ and
appropriately modifying the normal forms, we obtain an analogous
presentation of the symmetric monoidal groupoid of $\CNOT$+$T$
operators.

Our contributions can be seen as the reformulation of prior results in
the graphical language of quantum circuits. Indeed, it was shown in
\cite{CMBSG16} that the group of $n$-qubit $\CNOT$-dihedral operators
is isomorphic to the semidirect product $M\rtimes GA(n,\Z_2)$ where
$M$ is some subgroup of $\Z_8^{2^n}$ and $GA(n,\Z_2)$ is the general
affine group of order $n$ over the two-element field. Independently,
it was shown in \cite{AM} that the group of $n$-qubit $\CNOT$+$T$
operators is isomorphic to the semidirect product
$M'\rtimes GL(n,\Z_2)$ where $M'$ is some subgroup of $\Z_8^{2^n-1}$
and $GL(n,\Z_2)$ is the general linear group of order $n$ over the
two-element field. Using these characterizations, normal forms for
$\CNOT$-dihedral and $\CNOT$+$T$ circuits were discussed in
\cite{CMBSG16} and \cite{CH161} respectively. In contrast, we give
finitely many relations which are sufficient to generate all circuit
identities over $\s{\CNOT, T, X}$. Circuit transformations can
therefore take place at the circuit level which alleviates the need to
translate to and from another formalism. Moreover, a self-contained
equational theory of circuits is significantly easier to extend to new
gate sets since all equations remain valid in the presence of
additional gates.

The paper is organized as follows. In \cref{sec:prelims}, we discuss
preliminaries. In \cref{sec:gen-rel}, we introduce generators and
relations for $\CNOT$-dihedral operators. In \cref{sec:nfs}, we define
normal forms for $\CNOT$-dihedral circuits. In \cref{sec:existence},
we use the relations to show that every $\CNOT$-dihedral operator
admits a normal form. We show that distinct normal forms correspond to
distinct operators in \cref{sec:uniqueness}. Finally, we conclude and
discuss generalizations and future work in \cref{sec:conc}.

% --------------------------------------------------------------------
\section{Preliminaries}
\label{sec:prelims}

The notion of \emph{presentation} used here is similar to the usual
one used in group theory but applied to a more general algebraic
structure called a \emph{symmetric monoidal groupoid}. Working with
monoidal groupoids allows us to account for the usual horizontal
composition of unitaries (matrix multiplication) as well as for their
vertical composition (tensor product). In much the same way that a
presentation of a group implicitly provides the relations axiomatizing
the group operation, a presentation of a symmetric monoidal groupoid
implicitly includes relations which account for the horizontal and
vertical compositions and their interplay. We state these
\emph{structural rules} below in the graphical language of
circuits. For further details about symmetric monoidal groupoids, the
reader is encouraged to consult \cite{Mac, Sel2009}.

For every pair of operators $f$ and $g$ we have
\[
  \bifunctoriality
\]
The above equality is known as the \emph{bifunctorial law}. It implies
that circuits on disjoint sets of qubits commute and guarantees that
the collection of circuits under consideration forms a monoidal
groupoid. One obtains a symmetric monoidal groupoid in the presence of
a \emph{symmetry} which is a family of self-inverse operators which
act as generalized $\SWAP$ gates. For example, two instances of the
symmetry are
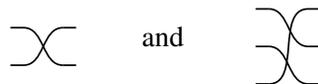
\[
  \symmetries
\]
which have the effect of permuting the order of the qubits. Every
instance of the symmetry satisfies a \emph{naturality law}, which
means that the symmetry has no effect beyond reordering the
qubits. For the instances above, the naturality is expressed by the
following circuit equalities, where $f$, $g$ and $h$ are arbitrary,
\[
  \Naturality
\]
In particular, the following \emph{spatial law} is a consequence of
the naturality of the symmetry, where $\lambda$ is an arbitrary scalar
represented as a gate without input or output wires.
\[
  \spatiality
\]
The symmetry also satisfies a property known as \emph{coherence} which
asserts that two circuits made of symmetries and implementing the same
permutation of wires are equal, e.g.,
\[
  \coherence
\]
Using symmetric monoidal groupoids allows us to focus on properties
that are specific to $\CNOT$-dihedral operators and to abstract away
generic properties of quantum circuits. In particular, the
bifunctorial law and the existence of a symmetry satisfying naturality
and coherence are assumed and needn't be explicitly included in the
presentation.

% --------------------------------------------------------------------
\section{Generators and relations}
\label{sec:gen-rel}

We recall the definition of the standard generators for
$\CNOT$-dihedral operators and introduce two derived generators to
streamline the presentation.

\begin{definition}
  \label{def:gens}
  The \emph{generators} are the scalar $\omega = e^{i\pi/4}$ and the
  gates $X$, $T$, and $\CNOT$ defined below.
  \[
    \begin{array}{lcr}
      \m{\begin{qcircuit}[scale=0.39]
          \grid{2.8}{0}
          \Xgate{}{1.4}{0}
        \end{qcircuit}} & = & \begin{bmatrix}
        0 & 1 \\
        1 & 0
      \end{bmatrix}
    \end{array}
    \qquad
    \begin{array}{lcr}
      \m{\begin{qcircuit}[scale=0.39]
          \grid{2.8}{0}
          \Tgate{}{1.4}{0}
        \end{qcircuit}} & = & \begin{bmatrix}
        1 & 0 \\
        0 & \omega
      \end{bmatrix}
    \end{array}
    \qquad
    \begin{array}{lcr}
      \m{\begin{qcircuit}[scale=0.39]
          \grid{2.8}{0,1.6}
          \CX{}{1.4}{1.6}{0}
        \end{qcircuit}} & = & \begin{bmatrix}
        1 & 0 & 0 & 0 \\
        0 & 1 & 0 & 0 \\
        0 & 0 & 0 & 1 \\
        0 & 0 & 1 & 0
      \end{bmatrix}
    \end{array}
  \]
\end{definition}

\begin{definition}
  \label{def:derived-gens}
  The \emph{derived generators} are the gates $U$ and $V$ defined
  below.
  \[
      \m{\begin{qcircuit}[scale=0.39]
          \grid{2.8}{0,1.6}
          \Ugate{}{1.4}{0}{1.6}    
        \end{qcircuit}}
      = \m{\begin{qcircuit}[scale=0.39]
          \cleargate{$T$}{2.6,1.6}
          \grid{7.2}{0,1.6}
          \CX{}{1.4}{1.6}{0}
          \Tgate{}{3.6}{0} 
          \CX{}{5.8}{1.6}{0}
        \end{qcircuit}}
    \qquad
      \m{\begin{qcircuit}[scale=0.39]
          \grid{2.8}{0,1.6,3.2}
          \Vgate{}{1.4}{0}{1.6}{3.2}    
        \end{qcircuit}}
      = \m{\begin{qcircuit}[scale=0.39]
          \cleargate{$T$}{2.6,3.2}
          \grid{11.6}{0,1.6,3.2}
          \CX{}{1.4}{3.2}{1.6}
          \CX{}{3.6}{1.6}{0}
          \Tgate{}{5.8}{0} 
          \CX{}{8}{1.6}{0}
          \CX{}{10.2}{3.2}{1.6}
        \end{qcircuit}}
  \]
\end{definition}

In accordance with \cref{sec:prelims}, we assume that all symmetries
are given and we refer to any instance of the symmetry as a $\SWAP$
gate. Because they act as affine transformations on basis states, we
refer to $X$, $\CNOT$, and $\SWAP$ as \emph{affine} gates and by
extension to circuits using only affine gates as \emph{affine
  circuits}. Similarly, we refer to $\omega$, $T$, $U$, and $V$ as
\emph{diagonal} gates and to circuits using only diagonal gates as
\emph{diagonal circuits}. If $C$ is a $\CNOT$-dihedral circuit, we
write $W_C$ to denote the operator represented by $C$. Note that if
$C$ is diagonal (resp. affine) circuit, then $W_C$ is diagonal
(resp. affine).

\begin{definition}
  \label{def:rels}
  The \emph{relations} are given in \cref{fig:rels}. We refer to
  relations $R_1$ through $R_6$ as \emph{affine relations}, to
  relations $R_7$ through $R_{10}$ as \emph{diagonal relations}, and
  to relations $R_{11}$ through $R_{13}$ as \emph{commutation
    relations}.
\end{definition}

\begin{figure}
  \input{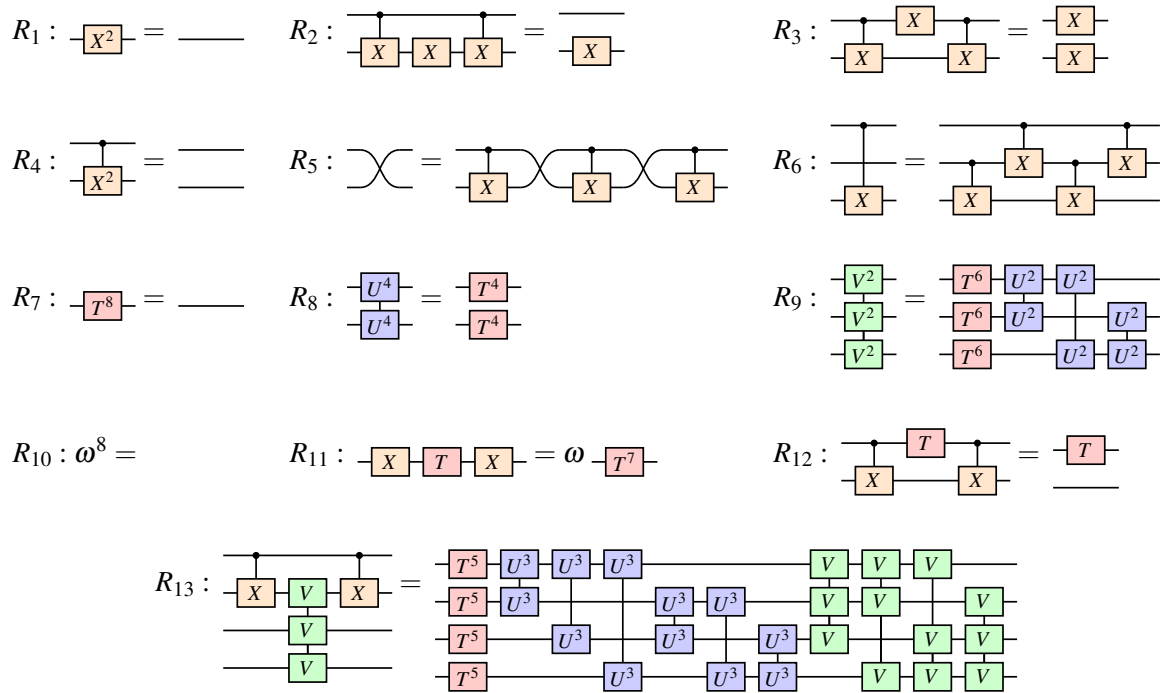}
  \caption{The relations. $R_1$ through $R_6$ are affine
    relations. $R_7$ through $R_{10}$ are diagonal relations. $R_{11}$
    through $R_{12}$ are commutation relations.}
  \label{fig:rels}
\end{figure}

In \cref{fig:rels} and throughout the rest of the paper, we use the
following notational conventions. We place global phases (i.e.,
scalars) in front of circuits as in the right-hand side of
$R_{11}$. Note that this is consistent with the spatial law. Gates
labelled $f^n$ for some integer $n\in\N$ denote the $n$-fold
composition of $f$ with itself, i.e.,
\[
  \block
\]
We call such a circuit an \emph{$f$-block of degree $n$}. By
extension, we write $\deg_f(C)$ for the maximum degree of the
$f$-blocks that appear in a circuit $C$.  We also use $\mathcal{T}$
and $\mathcal{X}$ to denote $T$ and $X$ blocks of arbitrary
degrees. Gates applied to non-adjacent qubits, as in $R_6$, $R_9$, and
$R_{13}$, are defined as adjacent-qubit gates on the top-most wires
conjugated by $\SWAP$ gates, e.g.,
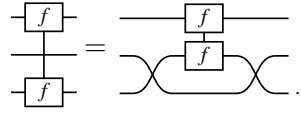
\[
  \NonAdjacent
\]
Because diagonal gates on non-adjacent qubits are diagonal in the
computational basis we mildly abuse terminology and refer to circuits
such as the right-hand side of $R_{13}$ as diagonal circuits, even if
they contain non-diagonal $\SWAP$ gates.

The 13 relations of \cref{fig:rels} can be verified by explicit
computation. However, it is more illuminating to use the formalism of
\emph{phase polynomials} introduced in \cite{AMMR}.

\begin{definition}
  \label{def:litteral}
  Let $\oplus$ denote addition in $\Z_2$ and $\overline{x}$ denote the
  complement of $x$ in $\Z_2$. A \emph{literal} $l$ is either a
  Boolean variable $x$ or its inverse $\overline{x}$. A \emph{term}
  over $n$ Boolean variables is an expression of the form
  $l_1 \oplus \ldots \oplus l_n$ where each $l_i$ is a literal.
\end{definition}

\begin{definition}
  \label{def:ppolys}
  The $T$, $X$, and $\CNOT$ gates act on basis states as
  $T\ket{x} = \omega^{x}\ket{x}$, $X\ket{x} = \ket{\overline{x}}$, and
  $\CNOT\ket{x_1x_2}=\ket{x_1(x_1\oplus x_2)}$. It follows that the
  action of a $\CNOT$-dihedral circuit $C$ on an arbitrary basis state
  is given by
  \begin{equation}
    \label{eq:ppoly-def1}
    W_C\ket{x_1x_2\cdots x_n}=\omega^{p_C(x_1, x_2,\dots, x_n)}
    \ket{f_C(x_1,x_2,\dots, x_n)},
  \end{equation}
  where $f_C:\Z_2^n \to \Z_2^n$ is an affine reversible operator and
  $p_C:\Z_2^n \to \Z_2$ is an expression of the form
  \begin{equation}
    \label{eq:ppoly-def2}
    p_C(x_1,\ldots, x_n) = \sum_{k}^{i=1}a_i \cdot g_i
  \end{equation}
  for some $k\in \N$, some $a_i\in \Z_8$, and some terms $g_i$ on no
  more than $n$ variables. The expression in \cref{eq:ppoly-def1} is
  the \emph{phase polynomial representation of $C$} and the one in
  \cref{eq:ppoly-def2} is the \emph{phase polynomial associated with
    $C$}.
\end{definition}

Note that the phase polynomials of \cref{eq:ppoly-def2} use mixed
arithmetic: the ``outside'' sum is computed modulo 8 since
$\omega^8=1$ while the ``inside'' sums are computed modulo 2. As with
usual polynomials, we write 0 for the phase polynomial whose
coefficients are all 0.

Phase polynomials are a concise representation of the action of
$\CNOT$-dihedral circuits on states and can be used to prove that two
distinct circuits represent the same operator.

\begin{proposition}
  \label{prop:soundness}
  The relations of \cref{fig:rels} are sound.
\end{proposition}

\begin{proof}
  We briefly discuss the case of $R_{13}$. Let $C_L$ and $C_R$ be the
  circuits on the left-hand side and right-hand side of $R_{13}$
  respectively. The phase polynomial representation of $C_L$ is
  \[
    W_{C_L}\ket{x_1x_2x_3x_4} = \omega^{x_1\oplus x_2\oplus x_3 \oplus
      x_4} \ket{x_1x_2x_3x_4}.
  \]
  It can be verified (see, e.g., \cite{AM}) that, for any
  $x_1, x_2, x_3, x_4 \in \Z_2$,
  \[
    x_1\oplus x_2\oplus x_3 \oplus x_4 = \sum_{i} 5x_i + \sum_{i<j}
    3(x_i\oplus x_j) + \sum_{i<j<k} x_i\oplus x_j\oplus x_k \mod 8.
  \]
  Hence
  \[
    W_{C_L}\ket{x_1x_2x_3x_4}=\omega^{\sum_{i} 5x_i + \sum_{i<j}
      3(x_i\oplus x_j) + \sum_{i<j<k} x_i\oplus x_j\oplus x_k}
    \ket{x_1x_2x_3x_4}
  \]
  which is the phase polynomial representation of $C_R$ so that
  $W_{C_L} = W_{C_R}$.
\end{proof}

\begin{remark}
  \label{rem:cnott-pres1}
  By considering only the $T$, $U$, $V$ and $\CNOT$ gates as
  generators, and by omitting the relations $R_1$, $R_2$, $R_3$,
  $R_{10}$, and $R_{11}$, one obtains a presentation of the symmetric
  monoidal groupoid of $\CNOT$+$T$ operators.
\end{remark}

\begin{remark}
  \label{rem:independence}
  It can be shown that the relations given in \cref{fig:rels} are
  \emph{independent}, in that it is not possible to derive one from
  the others. However, it is not currently known whether the relations
  are \emph{minimal}, though we believe this to be the case.
\end{remark}

% --------------------------------------------------------------------
\section{Normal forms}
\label{sec:nfs}

For each $\CNOT$-dihedral operator $W$ we choose a distinguished
circuit which we call the \emph{normal form} of $W$. We define normal
forms for affine and diagonal operators independently. For affine
operators, we use the normal forms introduced by Lafont in \cite{Laf}
which we recall here for completeness. In both cases, we introduce
convenient shorthand prior to introducing normal forms.

\begin{definition}
  \label{def:stairs}
  \emph{Ascending stairs} are circuits of the form
  \[
    \stair
  \]
  The identity circuit is the only ascending stair on a single
  qubit. \emph{Descending stairs} are defined similarly.
\end{definition}

\begin{definition}
  \label{def:ladder}
  \emph{Ladders} are circuits of the form
  \[
    \ladder
  \]
  The identity circuit is the only ladder on a single qubit.
\end{definition}

In \cref{def:ladder}, the ascending stair rises from the bottom qubit
to the top one. The descending stair, however, may or may not fall all
the way to the bottom qubit.

\begin{definition}
  \label{def:nf-affine}
  An \emph{affine normal form} is a circuit $A$ of the form
  \[
    \nfaff
  \]
  such that $\deg_X(A)\in\Z_2$ and $\deg_{\CNOT}(A)\in\Z_2$.
\end{definition}

\begin{example}
  \label{ex:affine}
  The affine operator defined by
  $\ket{x_1x_2x_3}\mapsto \ket{(\overline{x_2\oplus x_3})x_1(x_1\oplus
    x_2)}$, where $\oplus$ is addition in $\Z_2$ and $\overline{x}$ is
  the additive inverse of $x$ in $\Z_2$, has the following affine
  normal form
  \[
    \exampleAff
  \]
\end{example}

\begin{remark}
  \label{rem:order-affine}
  There are $2^{n-1}$ distinct stairs on $n$ qubits and thus $2^n - 1$
  distinct stairs on no more than $n$ qubits. This implies that the
  number of $n$-qubit ladders is $2^{n-1}(2^n-1)$ which in turn
  implies that the number of distinct affine normal forms on $n$
  qubits is
  \begin{equation}
    \label{eq:aff-nf}
    2^n \cdot \prod_{i=1}^n 2^{i-1}(2^i-1) = 2^n \cdot \prod_{i=1}^n
    (2^n -2^{i-1}).
  \end{equation}
  In \cref{eq:aff-nf}, the prefactor of $2^n$ accounts for the layer
  of $X$ gates which appear at the right of the normal form. Note that
  the expression in \cref{eq:aff-nf} coincides with the well-known
  formula for the cardinality of the general affine group of order
  $n$.
\end{remark}

\begin{definition}
  \label{def:Utriangle}
  \emph{$U$-triangles} are circuits of the form
  \[
    \Utriangle
  \]
  The identity circuit is the only $U$-triangle on a single qubit.
\end{definition}

\begin{definition}
  \label{def:Vtrapezoidtriangle}
  \emph{$V$-trapezoids} and \emph{$V$-triangles} are circuits of the
  form
  \[
    \Vtrapezoid \qquad \mbox{and} \qquad \Vtriangle
  \]
  The identity circuit is the only $V$-trapezoid or $V$-triangle on a
  single qubit. Similarly, the only $V$-triangle or $V$-trapezoid on
  two qubits is the identity circuit.
\end{definition}

\begin{definition}
  \label{def:nf-diag}
  A \emph{diagonal normal form} is a circuit $D$ of the form
  \[
    \nfdiag
  \]
  such that $k\in\Z_8$, $\deg_T(D)\in\Z_8$, $\deg_U(D)\in\Z_4$, and
  $\deg_V(D)\in\Z_2$.
\end{definition}

The normal forms introduced in \cref{def:nf-diag} correspond to an
ordering of the gates in a diagonal circuit according to which powers
of $\omega$ appear first, followed by $T$, $U$, and $V$ gates. The $U$
gates are positioned in lexicographical order, with respect to the set
of qubits they act on. The placement of $V$ gates also follows the
lexicographical ordering.

\begin{example}
  \label{ex:CCZ}
  The doubly-controlled Pauli $Z$ gate, whose matrix is
  $\diag(1,1,1,1,1,1,1,-1)$, has the following diagonal normal form
  \[
    \nfCCZ
  \]
\end{example}

\begin{remark}
  \label{rem:order-diag}
  In analogy with \cref{rem:order-affine}, we note that there are
  $8 \cdot 8^{\binom{n}{1}}\cdot 4^{\binom{n}{2}}\cdot
  2^{\binom{n}{3}}$ distinct diagonal normal forms.
\end{remark}

\begin{definition}
  \label{def:nf}
  A \emph{normal form} is a circuit of the form $DA$ where $D$ is a
  diagonal normal form and $A$ is an affine normal form.
\end{definition}

\begin{remark}
  \label{rem:order}
  It follows from \cref{rem:order-affine} and \cref{rem:order-diag},
  that the number of normal forms is
  \[
    8 \cdot 8^{\binom{n}{1}}\cdot 4^{\binom{n}{2}}\cdot
    2^{\binom{n}{3}} \cdot 2^n\cdot \prod_{i=1}^{n} (2^n-2^{i-1}) =
    2^{3 + 4\binom{n}{1} + 2\binom{n}{2} +
      \binom{n}{3}}\prod_{i=1}^{n} (2^n-2^{i-1}).
  \]
\end{remark}

\begin{remark}
  \label{rem:cnott-pres2}
  In the case of $\CNOT$+$T$ operators, the affine normal forms are
  replaced with \emph{linear} normal forms which are obtained by
  removing the final column of $X$-blocks from the circuits of
  \cref{def:nf-affine}. The $\CNOT$+$T$ diagonal normal forms are the
  scalar-free versions of the circuits of \cref{def:nf-diag}.
\end{remark}

% --------------------------------------------------------------------
\section{Existence}
\label{sec:existence}

In this section, we use the relations of \cref{fig:rels}, together
with the structural rules of \cref{sec:prelims}, to show that every
$\CNOT$-dihedral operator admits a normal form. For this, we first
establish that every $\CNOT$-dihedral circuit can be written as a
diagonal circuit, followed by an affine one. We then consider the
existence of diagonal and affine normal forms independently.

\begin{lemma}
  \label{lem:DA-decomp}
  If $C$ is a $\CNOT$-dihedral circuit, then there exists a diagonal
  circuit $D$ and an affine circuit $A$ such that $C=DA$.
\end{lemma}

\begin{figure}
  \input{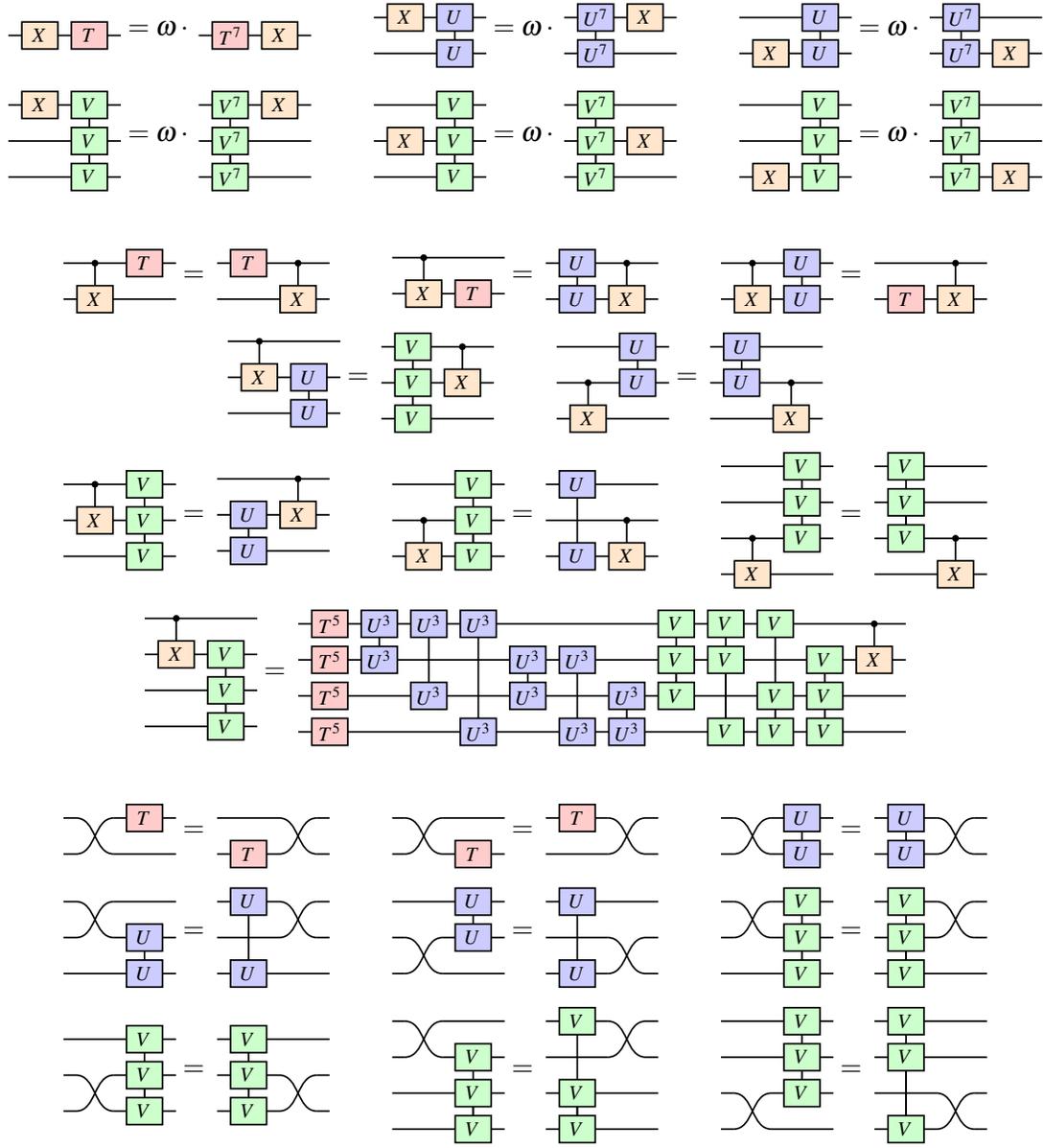}
  \caption{Derivable commutation rules.}
  \label{fig:DA-cases}
\end{figure}

\begin{proof}
  It suffices to show that the lemma is true when $C$ consists of a
  diagonal gate $d$ appearing to the right of an affine gate $a$. If
  $d$ and $a$ act on distinct qubits, they can be commuted by the
  bifunctorial law. Likewise, if $d$ is a power of $\omega$, it can be
  commuted past $a$ by the spatial law. This leaves the 25 cases
  listed in \cref{fig:DA-cases}. The first six cases show how to
  commute an $X$ gate past a diagonal gate. The next six cases show
  how to commute a $\CNOT$ gate past a diagonal gate. The last six
  cases show how to commute a $\SWAP$ gate past a diagonal
  gate. Verifying that each of these equations follows from the
  relations of \cref{fig:rels} is a tedious but straightforward
  exercise. We give an example derivation, using the relations $R_4$,
  $R_5$, $R_{12}$ and the bifunctorial law: \DerivationOne Note that
  in the last six cases of \cref{fig:DA-cases}, we only consider the
  two-qubit $\SWAP$, as opposed to more general $\SWAP$ gates. This is
  because coherence guarantees that an arbitrary $\SWAP$ can be
  expressed as a sequence of two-qubit $\SWAP$ gates.
\end{proof}

In order to prove that diagonal circuits admit a normal form, we start
by showing that commutation rules between diagonal gates can be
derived from the relations of \cref{fig:rels}.

\begin{lemma}
  \label{lem:diagonal-commute}
  Diagonal gates commute.
\end{lemma}

\begin{proof}
  As in the proof of \cref{lem:DA-decomp}, we only need to consider
  the cases where two diagonal gates act on at least one common
  qubit. Moreover, since $U$ and $V$ are symmetric with respect to the
  qubits they act on -- i.e., they commute with $\SWAP$ gates as shown
  in \cref{fig:DA-cases} -- we can further reduce the number of cases
  to consider to the following four.
  \[
  \m{\begin{qcircuit}[scale=0.38]
    \grid{5}{0,1.6}
    \Tgate{}{1.4}{1.6}
    \Ugate{}{3.6}{0}{1.6}
  \end{qcircuit}}
= \m{\begin{qcircuit}[scale=0.38]
    \grid{5}{0,1.6}
    \Ugate{}{1.4}{0}{1.6}
    \Tgate{}{3.6}{1.6}
  \end{qcircuit}} \quad
  \m{\begin{qcircuit}[scale=0.38]
    \grid{5}{0,1.6,3.2}
    \Tgate{}{1.4}{3.2}
    \Vgate{}{3.6}{0}{1.6}{3.2}
  \end{qcircuit}}
= \m{\begin{qcircuit}[scale=0.38]
    \grid{5}{0,1.6,3.2}
    \Vgate{}{1.4}{0}{1.6}{3.2}
    \Tgate{}{3.6}{3.2}
  \end{qcircuit}} \quad
  \m{\begin{qcircuit}[scale=0.38]
    \grid{5}{0,1.6,3.2}
    \Ugate{}{1.4}{1.6}{3.2}
    \Vgate{}{3.6}{0}{1.6}{3.2}
  \end{qcircuit}}
= \m{\begin{qcircuit}[scale=0.38]
    \grid{5}{0,1.6,3.2}
    \Vgate{}{1.4}{0}{1.6}{3.2}
    \Ugate{}{3.6}{1.6}{3.2}
  \end{qcircuit}} \quad
  \m{\begin{qcircuit}[scale=0.38]
    \grid{5}{0,1.6,3.2,4.8}
    \Ugate{}{1.4}{3.2}{4.8}
    \Vgate{}{3.6}{0}{1.6}{3.2}
  \end{qcircuit}}
= \m{\begin{qcircuit}[scale=0.38]
    \grid{5}{0,1.6,3.2,4.8}
    \Vgate{}{1.4}{0}{1.6}{3.2}
    \Ugate{}{3.6}{3.2}{4.8}
  \end{qcircuit}}
  \]
  Verifying that the above equations follow from the relations in
  \cref{fig:rels} is another straightforward exercise. As an example,
  we derive the fourth equation below, using the definition of $V$ as
  well as the fact that $U$ commutes with the top wire of the $\CNOT$
  gate, which is one of the derivable rules of \cref{fig:DA-cases}.
  \DerivationTwo
\end{proof}

\begin{lemma}
  \label{lem:nf-diag}
  Every diagonal circuit admits a normal form.
\end{lemma}

\begin{proof}
  Let $C$ be an $n$-qubit diagonal circuit. By
  \cref{lem:diagonal-commute} the gates of $C$ may be reordered into
  the form of \cref{def:nf-diag}. It therefore suffices to bound the
  degree of $T$, $U$ and $V$ blocks by $7, 3$ and $1$,
  respectively. We first reduce the degree of all $V$ blocks modulo
  $2$ by applying relation $R_9$. Note that the right hand side of
  $R_9$ contains no $V$ gates and hence does not increase the degree
  of any $V$ block. Once all $V$ blocks have been reduced and gates
  have been reordered and combined appropriately, the $U$ blocks may
  likewise be reduced modulo $4$ via $R_8$. Again, the right hand side
  of $R_8$ contains only $T$ gates and hence cannot increase the
  degree of any $U$ or $V$ blocks.  Finally $R_7$ may be used to
  reduce the remaining $T$ blocks to degree at most $7$.
\end{proof}

\cref{lem:nf-diag} establishes that diagonal circuits admit normal
forms. To prove that arbitrary $\CNOT$-dihedral circuits can be
normalized, we need an analogous result for affine circuits, which was
proved by Lafont.

\begin{lemma}[Lafont \cite{Laf}]
  \label{lem:nf-affine}
  Every affine circuit admits a unique normal form.
\end{lemma}

\begin{proposition}
  \label{prop:nf-existence}
  Every $\CNOT$-dihedral circuit admits a normal form.
\end{proposition}

\begin{proof}
  Let $C$ be a $\CNOT$-dihedral circuit. By \cref{lem:DA-decomp}, $C$
  can be written as a product $DA$ where $D$ is a diagonal circuit and
  and $A$ is an affine circuit. By \cref{lem:nf-diag} and
  \cref{lem:nf-affine}, $D$ has a diagonal normal form $D'$ and $A$
  has an affine normal form $A'$. The $\CNOT$-dihedral circuit $C$
  therefore has normal form $C=D'A'$.
\end{proof}

\begin{remark}
  \label{rem:cnott-existence}
  The existence of $\CNOT$+$T$ normal forms can be established by
  reasoning as in \cref{prop:nf-existence}.
\end{remark}

% --------------------------------------------------------------------
\section{Uniqueness}
\label{sec:uniqueness}

In this section, we show that normal forms are unique: distinct normal
forms represent distinct operators. To this end, we use the formalism
of phase polynomials introduced in \cref{sec:gen-rel}.

The action of powers of the diagonal gates of
definitions~\ref{def:gens} and \ref{def:derived-gens} on basis states
are $\omega^k\ket{x_1} = \omega^k\ket{x_1}$,
$T^k\ket{x_1}=\omega^{kx_1}\ket{x_1}$,
$U^k\ket{x_1x_2}=\omega^{k(x_1\oplus x_2)}\ket{x_1x_2}$, and
$V^k\ket{x_1x_2x_3}=\omega^{k(x_1\oplus x_2\oplus
  x_3)}\ket{x_1x_2x_3}$. As a result, if $D$ is a diagonal normal form
on $n$ qubits and $\ket{x}=\ket{x_1\ldots x_n}$ is a basis state then
$D\ket{x} = \omega^{p_D(x)}\ket{x}$, where $p_D(x)$ is an expression
of the form
\begin{equation}
  \label{eq:pp}
  p_D(x)= a_0 + \sum_i a_i\cdot x_i + 
  \sum_{i < j} b_{i,j}\cdot(x_i \oplus x_j) + 
  \sum_{i < j < k} c_{i,j,k}(x_i \oplus x_j \oplus x_k)
\end{equation}
with $a_i\in\Z_8$ for $i\in\s{0,\ldots,n}$, and $b_{i,j}\in\Z_4$,
$c_{i,j,k}\in\Z_2$, for $i,j,k \in \s{1,\ldots,n}$. Further, every
diagonal normal form corresponds to a unique expression of the form
\cref{eq:pp}, as distinct powers of the $T$, $U$, and $V$ gates
contribute to distinct terms in the expression $p_D(x)$.

To prove that every diagonal normal form represents a distinct
operator, it is helpful to express the mixed arithmetic polynomial of
\cref{eq:pp} as a \emph{multilinear polynomial over $\Z_8$}, i.e., as
a polynomial over $\Z_8$ that is linear in each of its variables
\cite{OD14}.

\begin{lemma}
  \label{lem:phase-to-multilinear}
  If $p(x)$ and $p'(x)$ are phase polynomials as in \cref{eq:pp} then
  there exists a multilinear polynomial $q(x)$ such that
  $p(y) - p'(y)= q(y)$ for all $y\in\Z_2^n$. Moreover, if
  $p(x)- p'(x) \neq 0$ then $q(x)\neq 0$.
\end{lemma}

\begin{proof}
  It can be verified by computation that the following equalities hold
  for $x_i,x_j,x_k\in\Z_2$
  \begin{align*}
    x_i\oplus x_j &= x_i + x_j - 2x_ix_j \\
    x_i\oplus x_j\oplus x_k &= x_i + x_j + x_k - 2x_ix_j - 2x_ix_k - 
                              2x_jx_k + 4x_ix_jx_k
  \end{align*}
  where $\oplus$ is addition in $\Z_2$ but all other arithmetic
  operations are performed in $\Z_8$. The first claim follows by
  applying the above equalities to $p(x) - p'(x)$. For the second
  claim, note that if $p(x)-p'(x)\neq 0$, we must have
  $a_i-a_i'\neq 0$ modulo 8, $b_{ij}-b_{ij}'\neq 0$ modulo 4, or
  $c_{ijk}-c_{ijk}'\neq 0$ modulo 2. If there exists $i,j,k$ such that
  $c_{ijk}-c_{ijk}'\neq 0$ modulo 2, then $4(c_{ijk}-c_{ijk}')\neq 0$
  modulo 8. This implies that $q(x)\neq 0$, since
  $4(c_{ijk}-c_{ijk}')$ is the unique coefficient associated with the
  monomial $x_ix_jx_k$. If no such $i,j,k$ exists, we can reason
  analogously with a coefficient of the form $b_{ij}-b_{ij}'$ or
  $a_i - a_i'$.
\end{proof}

\begin{lemma}
  \label{lem:diagonal-uniqueness}
  Distinct diagonal normal forms represent distinct operators.
\end{lemma}

\begin{proof}
  Let $D$ and $D'$ be distinct diagonal normal forms with phase
  polynomials $p_D(x)$ and $p_{D'}(x)$ respectively. Since $D$ and
  $D'$ are normal, $p_D(x)$ and $p_{D'}(x)$ are of the form given in
  \cref{eq:pp}. And since $D$ and $D'$ are distinct, $p_D(x)$ and
  $p_{D'}(x)$ are likewise distinct, hence $p_D(x) - p_{D'}(x)\neq
  0$. \cref{lem:phase-to-multilinear} therefore implies that there
  exists a nonzero multilinear polynomial $q(x)$ such that
  $p(y) - p'(y)= q(y)$ for all $y\in\Z_2^n$. Now let
  $d \cdot x_{i_1}\ldots x_{i_j}$ be a non-zero term in $q(x)$ of
  lowest degree and let $y\in\Z_2^n$ be the vector with $1$'s in the
  $i_1\cdots i_j$ positions and $0$'s elsewhere. Then
  $q(y) = d \neq 0$, which implies that $p_D(y)-p_{D'}(y)\neq 0$ and
  therefore that $D\ket{y}\neq D'\ket{y}$.
\end{proof}

\cref{lem:diagonal-uniqueness} establishes that diagonal normal forms
are unique. To obtain the uniqueness of normal forms, we need a
similar result for affine normal forms, which was proved by Lafont.

\begin{lemma}[Lafont \cite{Laf}]
  \label{lem:affine-uniqueness}
  Distinct affine normal forms represent distinct operators.
\end{lemma}

\begin{proposition}
  \label{prop:uniqueness}
  Distinct normal forms represent distinct operators.
\end{proposition}

\begin{proof}
  Let $C$ and $C'$ be two normal forms. By definition, $C=DA$ and
  $C'=D'A'$ for some diagonal normal forms $D,D'$ and some affine
  normal forms $A,A'$. Suppose that $C$ and $C'$ represent the same
  operator, i.e., that $W_C=W_{C'}$. Then $W_{A}W_{D} = W_{A'}W_{D'}$
  and therefore $W_{D} = W_{A}^\dagger W_{A'}W_{D'}$. Since $W_D$ and
  $W_{D'}$ are diagonal, $W_A^\dagger W_{A'}$ is a diagonal affine
  operator and thus $W_A^\dagger W_{A'}=1$, or $W_A=W_{A'}$. This
  implies that $W_D=W_{D'}$. The result then follows from
  \cref{lem:diagonal-uniqueness} and \cref{lem:affine-uniqueness}.
\end{proof}

By \cref{prop:uniqueness} and \cref{prop:nf-existence}, there is a
bijection between normal forms and $\CNOT$-dihedral operators so that
the number of $n$-qubit $\CNOT$-dihedral operators is equal to the
number of normal forms on $n$ qubits that was computed in
\cref{rem:order}.

\begin{corollary}
  \label{cor:cardinality}
  The order of the group of $\CNOT$-dihedral operators on $n$ qubits
  is
  \[
    2^{3 + 4\binom{n}{1} + 2\binom{n}{2} +
      \binom{n}{3}}\prod_{i=1}^{n} (2^n-2^{i-1}).
  \]
\end{corollary}

\begin{remark}
  \label{rem:cnott-uniqueness}
  The results of this section can be adapted to show that distinct
  $\CNOT$+$T$ normal forms represent distinct operators which then
  implies that the number of $\CNOT$+$T$ operators is
  \[
    2^{3\binom{n}{1} + 2\binom{n}{2} + \binom{n}{3}} \prod_{i=1}^{n}
    (2^n-2^{i-1}).
  \]
\end{remark}

\begin{remark}
  \label{rem:hadamard}
  The $\CNOT$-dihedral operators are not universal for quantum
  computation. One obtains the universal Clifford+$T$ gate set by
  adding the following Hadamard gate to the generators
  \[
    H = \frac{1}{\sqrt 2}\begin{bmatrix}
      1 & 1 \\
      1 & -1
    \end{bmatrix}.
  \]
  Since the Hadamard gate is not diagonal, one may wonder to what
  extent it contributes to diagonal Clifford+$T$ operators. We can use
  \cref{cor:cardinality} to quantify this contribution. Indeed, there
  are
  \[
    2^{3 + 4\binom{n}{1} + 2\binom{n}{2} + \binom{n}{3}} = O(2^{n^3})
  \]
  diagonal $\CNOT$-dihedral operators on $n$ qubits. In comparison, it
  is known from \cite{GS13} that for $n\geq 4$, the number of
  ancilla-free diagonal Clifford+$T$ operators on $n$ qubits is
  $8^{2^n-1} = O(2^{2^n})$. The Hadamard gate therefore contributes to
  the vast majority of diagonal Clifford+$T$ operators.
\end{remark}

% --------------------------------------------------------------------
\section{Conclusion}
\label{sec:conc}

We gave a finite presentation of the symmetric monoidal groupoid of
$\CNOT$-dihedral operators of order 16. To this end, we introduced a
notion of normal form for $\CNOT$-dihedral circuits and showed that
every $\CNOT$-dihedral operator admits a unique normal form. As a
corollary, we obtained a finite presentation of the symmetric monoidal
groupoid of $\CNOT$+$T$ operators.

Although we have shied from doing so in this paper, our methods can be
extended to $\CNOT$-dihedral operators of higher order. For
$\CNOT$-dihedral operators of order $2n$, the generators $\omega$ and
$T$ are replaced with the scalar $\zeta_n = e^{2\pi i / n}$ and the
phase gate
\[
  \begin{bmatrix}
    1 & 0 \\
    0 & \zeta_n
  \end{bmatrix}.
\]
A presentation may then be obtained by modifying the diagonal
relations appropriately. The results of \cite{AM} can be used to show
that it is sufficient to include the relevant order relations (akin to
$R_7$ and $R_{10}$) as well as relations reducing the order of
multi-qubit phase gates (akin to $R_8$, $R_9$, and $R_{13}$). In the
latter case, it suffices to introduce, for each $2^k$ dividing $n$, a
relation between a $k+1$ qubit phase gate of order $2^k$ and a circuit
using phase gates of smaller arity.

An avenue for future research is to find a rewrite system for
$\CNOT$-dihedral circuits. Indeed, \cref{prop:nf-existence}
establishes that every $\CNOT$-dihedral operator admits a normal form
but it does not contain an algorithm to normalize an arbitrary
$\CNOT$-dihedral circuit via rewriting. This is because the proof of
\cref{prop:nf-existence} appeals non-constructively to properties of
the ambient symmetric monoidal structure. Recent results in rewriting
theory address this problem \cite{BGKSZ} and might be used in order to
obtain an effective presentation of $\CNOT$-dihedral operators.

% --------------------------------------------------------------------
\section{Acknowledgements}
\label{sec:acknowledgements}

MA and NJR wish to thank the Banff International Research Station
(BIRS) where the ideas presented here were first discussed. NJR wishes
to thank Dmitri Maslov for sparking his interest in restricted
Clifford+$T$ circuits and to thank David Gosset and Yves Guiraud for
stimulating discussions. MA, JC, and NJR thank Miriam Backens and
anonymous referees for helpful comments on an earlier version of this
paper.

MA is partially funded by Canada's NSERC. JC and NJR are funded by the
Department of Defense.

% --------------------------------------------------------------------
\bibliographystyle{eptcs} \bibliography{cnotdihedral}

\end{document}